\documentclass[journal,a4paper]{IEEEtran}
\IEEEoverridecommandlockouts
\usepackage{siunitx}
\usepackage{cite}
\usepackage{amsmath,amssymb,amsfonts}
\usepackage{algorithmic}
\usepackage[linesnumbered,ruled,vlined]{algorithm2e}
\usepackage{mathtools}
\usepackage{amsthm}
\usepackage{tikz}
\usepackage{hyperref}
\usepackage{datetime}
\settimeformat{hhmmsstime}
\newcolumntype{H}{>{\setbox0=\hbox\bgroup}c<{\egroup}@{}}
\usepackage{xcolor}
\usepackage{multirow,booktabs}
\usepackage[utf8]{inputenc}
\usepackage{microtype}
\usepackage{tabularray}

\newcommand{\rmBtr}{\color{green!65!black}}
\newcommand{\rmWrs}{\color{red}}
\newcommand{\rmEql}{}
\newcommand{\rmUkn}{}

\DeclareMathOperator{\dist}{dist}
\DeclareMathOperator{\OPT}{OPT}

\newcommand{\Rset}{ {\mathbb{R}} }
\newtheorem{theorem}{Theorem}

\newtheorem{lemma}[theorem]{Lemma}

\usepackage{graphicx}
\usepackage{textcomp}
\def\BibTeX{{\rm B\kern-.05em{\sc i\kern-.025em b}\kern-.08em
    T\kern-.1667em\lower.7ex\hbox{E}\kern-.125emX}}

\begin{document}

\title{Cost-Distance Steiner Trees\\ for Timing-Constrained Global Routing}

\author{\IEEEauthorblockN{Stephan Held and Edgar Perner}\\
\IEEEauthorblockA{\textit{Research Institute for Discrete Mathematics and Hausdorff Center for Mathematics} \\
\textit{University of Bonn}\\
Bonn, Germany \\
\href{mailto:held@dm.uni-bonn.de}{held@dm.uni-bonn.de} and \href{mailto:perner@dm.uni-bonn.de}{perner@dm.uni-bonn.de}
}
}

\maketitle
\thispagestyle{plain}
\pagestyle{plain}

\begin{abstract}
The cost-distance Steiner tree problem seeks a Steiner tree that minimizes the total congestion cost plus the weighted sum of source-sink delays.
This problem arises as a subroutine in timing-constrained global routing with a linear delay model, used before buffer insertion.
Here, the congestion cost and the delay of an edge are essentially uncorrelated, unlike in most other algorithms for timing-driven Steiner trees.

We present a fast algorithm for the cost-distance Steiner tree problem.
Its running time is $\mathcal{O}(t (n\log n + m))$, where $t$, $n$, and $m$ are the numbers of terminals, vertices, and edges in the global routing graph.
We also prove that our algorithm guarantees an approximation factor of $\mathcal{O}(\log t)$.
This matches the best-known approximation factor for this problem, but with a much faster running time.

To account for increased capacitance and delays after buffering caused by bifurcations, we incorporate a delay penalty for each bifurcation without compromising the running time or approximation factor.

In our experimental results, we show that our algorithm outperforms previous methods that first compute a Steiner topology, e.g. based on shallow-light Steiner trees or the Prim-Dijkstra algorithm, and then embed this into the global routing graph.
\end{abstract}


\section{Introduction}

Modern routers do not only aim for minimum total wire length, but need
to consider signal timing as well \cite{HuSapatnekar:02, Albrecht-etal:03, Yan-etal:06,Held-etal:18,Yang=etal:23,GlobalInterconnectOpt:23,Yang-etal:24}. Before  and in preparation of buffer insertion, this is
usually estimated using a linear delay model, where the delay depends
linearly on the wire length, as well as on wire type and layer \cite{Bartoschek-etal:10,CATALYST:13}.
There are many Steiner tree variants that take path lengths into
account \cite{PrimDijkstra:95,Bartoschek-etal:10,rotter-held:13,PrimDijkstra:18,SALT:19}.
They do not consider independent congestion and delay costs.

The cost-distance Steiner tree problem overcomes this limitation.
It asks for a Steiner tree minimizing the  sum of routing congestion cost
and weighted path delays through the tree.
It occurs as the central  sub-problem in timing-constrained global routing,
where it is usually  solved heuristically \cite{Held-etal:18}.

We are given a graph $G$ with edge cost $c: E(G) \to \Rset_+$ and edge
delays $d: E(G) \to \Rset_+$. In our application
$G$ is a 3D global routing graph, where an edge cost $c(e)$ arises from the current edge
usage and $d(e)$ is the delay of $e$ in a linear delay model.  If
multiple wire types (width and spacing configurations) are available
$G$ may have a parallel edge for each wire type that has an individual
cost and delay.

Given a net connecting a set of sinks $S$ to a root $r$, delay weights $w:\to \Rset_+$ and
positions $\pi: S\cup\{r\}\to V(G)$, we seek  for a Steiner tree $T$ connecting $S\cup\{r\}$ that is embedded into
$G$, i.e.\ has an extension  $\pi:V(T)\setminus (S\cup\{r\}) \to V(G)$
and minimizes

\begin{align}
  \label{eqn:cost-distance-problem}
  \text{cost}(T)=\sum _{e\in T} c(e) +\sum _{t \in S} w(t) \mathrm{delay}_T(r,t),
\end{align}

where $\mathrm{delay}_T(r,t) = \sum_{e\in E(T_{[r,t]})} d(e)$ in a simple linear delay model.
The delay weights indicate the criticality of a sink and arise from a
Lagrangean relaxation in timing-constrained global routing \cite{Held-etal:18}.

Throughout this paper, we use $t\coloneqq |S\cup\{r\}|, n\coloneqq |V(G)|$ and $m\coloneqq |E(G)|$.
The best known approximation factor for this problem is $\mathcal{O}(\log t)$ 
based on iteratively computing a matching of terminals that are served by a Steiner vertex based on a special distance function \cite{meyerson-etal:08,Chekuri-et-al01} with a running time of $\mathcal{O}(t^2(n\log n + m))$.
It does not permit an approximation factor better than $\Omega(\log \log t)$ unless $\text{NP} \subseteq\text{DTIME}(t^{\mathcal{O}(\log\log\log t)})$ \cite{Chuzhoy-etal:08}.
For the uniform case, where $c=d$, constant approximation factors are possible \cite{meyerson-etal:08,UniformCostDist:23}.

Bifurcations on a path increase the capacitance and delay after buffering.
In Figure \ref{fig:bifurcations}, we see a critical path with many bifurcations on the left, which
is unfavorable to the tree on the right after buffering.  To mitigate this,
we use a delay penalty for each bifurcation, as proposed in \cite{Bartoschek-etal:10}.
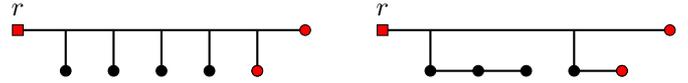
\begin{figure}[tb]
\begin{tikzpicture}[scale = 0.9, xscale = 0.7]
	\draw[thick] (0,0) -- (6,0);

	\foreach \x in {1,2,3,4,5} {
	\node[draw,circle,fill=black,minimum size=4pt,inner sep=0pt] at (\x,-0.6) {};  
	\draw[thick] (\x,0) -- (\x,-0.6) node[below] {};
	}

	\node[draw,rectangle,fill=red,minimum size=4pt,inner sep=0pt,label=above:$r$] at (0,0) {};
	\node[draw,circle,fill=red,minimum size=4pt,inner sep=0pt] at (6,0) {};
	\node[draw,circle,fill=red,minimum size=4pt,inner sep=0pt] at (5,-0.6) {};
	
\end{tikzpicture} \hfill
\begin{tikzpicture}[scale = 0.9, xscale = 0.7]
	\draw[thick] (0,0) -- (6,0);

        \foreach \x in {1,2,3,4,5} {
	\node[draw,circle,fill=black,minimum size=4pt,inner sep=0pt] at (\x,-0.6) {};  
	}
	\draw[thick] (1,0) -- (1,-0.6) node[below] {};
	\draw[thick] (4,0) -- (4,-0.6) node[below] {};
	\draw[thick] (3,-0.6) -- (1,-0.6) node[below] {};
	\draw[thick] (5,-0.6) -- (4,-0.6) node[below] {};

	\node[draw,rectangle,fill=red,minimum size=4pt,inner sep=0pt,label=above:$r$] at (0,0) {};
	\node[draw,circle,fill=red,minimum size=4pt,inner sep=0pt] at (6,0) {};
	\node[draw,circle,fill=red,minimum size=4pt,inner sep=0pt] at (5,-0.6) {};
	
\end{tikzpicture}
\caption{Two trees for the same net, the right one with fewer bifurcations on the path from the root $r$ to critical sinks (red dots).}
\label{fig:bifurcations}
\end{figure}

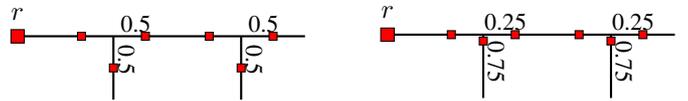
\begin{figure}[tb]
\begin{tikzpicture}[scale = 0.42]
	\draw[thick] (2,0) -- (11,0);
	\node[draw,rectangle,fill=red,minimum size=3pt,inner sep=0pt] at (4,0) {};
	\node[draw,rectangle,fill=red,minimum size=3pt,inner sep=0pt] at (6,0) {};
	\node[draw,rectangle,fill=red,minimum size=3pt,inner sep=0pt] at (8,0) {};
	\node[draw,rectangle,fill=red,minimum size=3pt,inner sep=0pt] at (10,0) {};
	\node[draw,rectangle,fill=red,minimum size=5pt,inner sep=0pt,label=above:$r$] at (2,0) {};
	\draw[thick] (5,0) -- (5, -2) {};
	\node[draw,rectangle,fill=red,minimum size=3pt,inner sep=0pt] at (5,-1) {};
	\draw[thick] (9,0) -- (9, -2) {};
	\node[draw,rectangle,fill=red,minimum size=3pt,inner sep=0pt] at (9,-1) {};

	\node at (5.7,0.4) {\small0.5};
	\node at (9.7,0.4) {\small0.5};
   \node[label={[rotate=-90]left:\small 0.5}] at (5.7,-1.5) {};
   \node[label={[rotate=-90]left:\small0.5}] at (9.7,-1.5) {};
\end{tikzpicture} \hfill
\begin{tikzpicture}[scale = 0.42]
	\draw[thick] (2,0) -- (11,0);
	\node[draw,rectangle,fill=red,minimum size=3pt,inner sep=0pt] at (4,0) {};
	\node[draw,rectangle,fill=red,minimum size=3pt,inner sep=0pt] at (6,0) {};
	\node[draw,rectangle,fill=red,minimum size=3pt,inner sep=0pt] at (8,0) {};
	\node[draw,rectangle,fill=red,minimum size=3pt,inner sep=0pt] at (10,0) {};
	\node[draw,rectangle,fill=red,minimum size=5pt,inner sep=0pt,label=above:$r$] at (2,0) {};
\draw[thick] (5,0) -- (5, -2) {};
	\node[draw,rectangle,fill=red,minimum size=3pt,inner sep=0pt] at (5,-0.2) {};
	\draw[thick] (9,0) -- (9, -2) {};
	\node[draw,rectangle,fill=red,minimum size=3pt,inner sep=0pt] at (9,-0.2) {};

					 \node at (5.7,0.4) {\small0.25};
					 \node at (9.7,0.4) {\small0.25};
					 \node[label={[rotate=-90]left:\small0.75}] at (5.7,-1.8) {};
					 \node[label={[rotate=-90]left:\small0.75}] at (9.7,-1.8) {};

\end{tikzpicture}
\caption{Two  buffering solutions (small red squares) indicating  different delay tradeoffs between at the branchings.}
\label{fig:distribute}
\end{figure}
As buffering has some freedom to shield off capacitance (see Figure~\ref{fig:distribute}), we use a model where
the bifurcation penalty can be distributed within certain limits. 
For technical reasons, we assume that the root and the sinks form the leaves of $T$ and internal vertices have degree at most 3. We say that  such a tree is \textit{bifurcation compatible}.
As we allow multiple vertices with the
same position, any Steiner tree can be transformed into such a tree without changing the total length
or any source-sink length in the tree.

We can consider  $T$ as an $r$-arborescence.
Then, at every proper Steiner vertex $u$ with $\delta^+_T(u) = \{x,y\}$ we distribute a delay
penalty to $(u,x)$ and $(u,y)$ as follows.  Let $d_{bif}\in \Rset_+$
be the total bifurcation penalty for both branches and $0 \le \eta \le 1/2$.
For $\lambda_x\in [\eta, 1-\eta]$, we assign a delay penalty
  of $\lambda_x d_{bif}$ to $(u,x)$ and $\lambda_yd_{bif}$ to
  $(u,y)$ with $\lambda_y = (1-\lambda_x)$.

 In the cost-distance
  objective with weighted delays, the optimum choice of $\lambda_x$  depends only on the
  total delay weights in the two subtrees:
  \begin{equation}
        \lambda_x =
  \begin{cases}
    \eta &    \text{if } w(S\cap V(T_x)) > w(S\cap V(T_y)) \\
    0.5 & \text{if } w(S\cap V(T_x)) = w(S\cap V(T_y)) \\
    (1-\eta) &  \text{else}.
  \end{cases}
  \label{eqn:lambda}
  \end{equation}
  If $w(S\cap V(T_x)) = w(S\cap V(T_y))$ any choice is optimal.
$T_x $ is the sub-arborescence rooted at $x$ and $w(X)=\sum_{x\in X}w(x)$.
  In \cite{Bartoschek-etal:10}, the special case $\eta = 0.5$ was considered.
  Minimizing the maximum delay would lead to different choices of $\lambda_x$.

  Setting $\lambda_x$  according to  (\ref{eqn:lambda}) for every edge $(u,x)$ with $|\delta^+(u)| = 2$
  and $\lambda_x = 0$ elsewhere, we have the following delay model, which we assume for (\ref{eqn:cost-distance-problem}) from now on:
\begin{equation}
    \mathrm{delay}_T(r,t) = \sum_{e=(u,v)\in E(T_{[r,t]})} \left(d(e) + \lambda_v d_{bif}\right).
    \label{eqn:bif-delay-formula}
\end{equation}

  The parameter $d_{bif}$ is computed similarly to
  \cite{Bartoschek-etal:10} by taking an optimally spaced uniform repeater
  chain.  Then, $d_{bif}$ is the delay increase when adding the input
  capacitance in the middle of a single net, minimizing over all layers
  and wire types.

\subsection{Our Contributions}
Our contributions are the following:
\begin{enumerate}
\item A new randomized  $\mathcal{O}(\log t)$-approximation algorithm for cost-distance Steiner trees with bifurcation penalties (Section~\ref{sec:approx-algorithm}). Its running time  of  $\mathcal{O}(t(n\log n + m))$ improves by a factor of $t$ over \cite{meyerson-etal:08}.
\item Multiple enhancements to improve its practical performance (Section~\ref{sec:pracical-enhancements}).
\item Experimental comparisons with the state-of the art heuristics, demonstrating a significant benefit
  of our method in timing-constrained global routing  (Section~\ref{sec:experimental-results}).
\end{enumerate}

\section{A Fast Approximation Algorithm}
 \label{sec:approx-algorithm}
Our algorithm (Algorithm~\ref{greedycd}) works similar to Kruskal's algorithm for spanning trees \cite{Kruskal:56}. It  builds a Steiner tree by iteratively merging components.
In each iteration $i$,
we propagate distance labels from all sinks simultaneously but with a special sink-specific distance function (see Theorem~\ref{thm:running-time}) until a permanent label reaches another terminal.
This determines the vertices $u$ and $v$ in Line~\ref{algline:dijkstra},
which are joined through a new Steiner point $s$ that replaces $u$ and $v$ in the terminal set $S_{i+1}$ of the next iteration and has the sum of their delay weights as its own weight.
Note that the component containing the root vertex $r_i$ might also be merged.
In this case, $s$ is  inserted at the position $\pi(r_i)$ and serves as the root $r_{i+1}$ in the next iteration (line~\ref{algline:rootmerge}).
If two sinks are merged, we select the position $\pi(s)$ randomly proportional to their delay weight (line~\ref{algline:sinkmerge}).
In the end, the inserted Steiner vertices are the  bifurcations of the tree. 

A key difference to \cite{Kruskal:56} is the cost function. In  \cite{meyerson-etal:08},
the main algorithm uses for each pair $u,v$ of sinks an individual distance function $c(e) + 2w(u)w(v)/(w(u)+w(v))\cdot d(e)$ ($e\in E(G)$). It combines costs and weighted delays.
Computing shortest paths between each pair of terminals with an individual cost function would be prohibitively slow in a large global routing graph.
We use one distance function
\begin{equation}
  l_u(e)\coloneqq c(e) + w(u) \cdot d(e) \quad (e\in E(G))
\end{equation}
per sink $u$ as in \cite{Chekuri-et-al01}.
While the algorithms in \cite{meyerson-etal:08,Chekuri-et-al01} merge iteratively based on a matching involving (a constant fraction of) all terminals,
our algorithm allows the newly inserted Steiner vertex to be merged in the next iteration, before any other merges take place.

We address the bifurcation delay penalties as follows.
Let 
$$\beta(w,w') = d_{bif}  \left(\eta \max(w,w') + (1 - \eta)\min(w, w') \right)$$
denote the minimum possible weighted delay penalty when merging two sinks with delay weights $w$ and $w'$. 
Then, the weighted bifurcation delay penalty for  merging  $u \in S_i$ and $v \in S_i\cup\{r_i\}$ is
$$
b(u,v) \coloneqq  
\begin{cases}
  \beta(w(u), w(v))                   & (v \in S_i)\\
  \beta(w(u), w(S_i \setminus \{u\})) & (v = r_i)
\end{cases}.
$$
In line \ref{algline:dijkstra} of Algorithm~\ref{greedycd}, the pair $u,v$ minimizing

\begin{align}
  L(u,v) \!\!\coloneqq \!\! \begin{cases}
    \dist_{G,c+ \min\{w(u),w(v)\} \cdot d}(u,v) + b(u,v)  &  \!\!(v\in S_i)\\
    \dist_{G,c+ w(u)\cdot d}(u,r_i) + b(u,r_i)  &  \!\!(v = r_i)
  \end{cases}
  \label{eqn:luv}
\end{align}

is chosen.
By $\dist_{\Gamma,\gamma}(x,y)$ we denote the shortest path distance between $x$ and $y$ in the graph $\Gamma$ w.r.t. edge lengths $\gamma$.
  $L(r_i,v)\coloneqq L(v,r_i)$ completes the definition of $L$.
  We can still prove an approximation factor of $\mathcal{O}(\log t)$.
In addition, our algorithm is  substantially faster than  \cite{meyerson-etal:08,Chekuri-et-al01}.

\begin{algorithm}[tb]
  	\caption{Cost-Distance Algorithm}\label{greedycd}
	\DontPrintSemicolon
	\KwIn{An instance $(G,S,r,w,c,d,d_{bif},\eta)$ of the Cost-Distance Steiner tree problem}
	\KwOut{An embedded Steiner tree $(T,\pi)$ on $ S $ with root $ r $.}
	$T \leftarrow (S \cup \{r\}, \emptyset)$, $i \leftarrow 0$, $S_0 \leftarrow S$, $r_0 \leftarrow r$.\;
	\While{$|S_i|>0$}{
	  Find the nodes $u \in S_i $, $v \in S_i \cup \{r_i\} $ minimizing $L(u,v)$ (see (\ref{eqn:luv})\& Theorem~\ref{thm:running-time}).\label{algline:dijkstra}\;
	\uIf{$v = r_i$} {
          Add a new Steiner vertex $s$ placed at  $\pi(s)= \pi(r)$, $r_{i+1}  \leftarrow s$, $w(s) \leftarrow w(u)$,
          $S_{i+1} \leftarrow S_{i} \setminus\{u\}$.\label{algline:rootmerge}
        } \Else {
          Randomly select a position $p \in \{\pi(u), \pi(v)\}$  with probabilities proportional to their delay weight.
          Add a new Steiner vertex $s$ at $\pi(s) = p$  with $w(s) \leftarrow w(u)+w(v)$,
          $S_{i+1} \leftarrow S_{i} \cup \{s\} \setminus\{u,v\}$.\label{algline:sinkmerge}
        }
	Add the $u$-$s$-$v$ connection corresponding to the path underlying $L(u,v)$ to $T$.\;
        $i \leftarrow i+1$.\;
	}
	Return the embedded tree $(T,\pi)$.\;
\end{algorithm}
\begin{figure}
	\begin{tikzpicture}[scale = 0.5485]
		\useasboundingbox (-1,-1) rectangle (7,5);
		\clip (-1,-1) rectangle (7,5);
		\draw[thick] (-1,-1) rectangle (7,5);
	
	\coordinate (P1) at (1, 0);
	\coordinate (P2) at (2, 2);
	\coordinate (P3) at (0, 2);
	\coordinate (P4) at (4, 0);
	\coordinate (P5) at (3, 4);
	\coordinate (P6) at (6, 0.1);
	
	\foreach \i/\r in {1/0.5,2/1,3/1.5,4/2,5/2.5} {
	  
	  \fill[black] (P\i) circle ({0.15/\r});
	}
	\fill[red] (P6) rectangle ++(0.2, -0.2) {};
	\node at (6.15,0.3) {$r$};
	\end{tikzpicture}
\begin{tikzpicture}[scale = 0.5485]
	\useasboundingbox (-1,-1) rectangle (7,5);
	\clip (-1,-1) rectangle (7,5);
	\draw[thick] (-1,-1) rectangle (7,5);

\coordinate (P1) at (1, 0);
\coordinate (P2) at (2, 2);
\coordinate (P3) at (0, 2);
\coordinate (P4) at (4, 0);
\coordinate (P5) at (3, 4);
\coordinate (P6) at (6, 0.1);
\coordinate (S) at (3, 2);
\draw[red] (P5) -- (S);
\draw[red] (P2) -- (S);

\foreach \i/\r in {1/0.5,2/1,3/1.5,4/2,5/2.5} {
  
  \fill[blue] (P\i) circle ({0.15/\r});
  \draw[blue] (P\i) circle ({0.9*\r});
}
\fill[red] (P6) rectangle ++(0.2, -0.2);
\node at (0,4.3) {\small $i=0$};
\node at (3.5,4) {\small $u$};
\node at (2,1.5) {\small $v$};
\end{tikzpicture}
\begin{tikzpicture}[scale = 0.5485]
	\useasboundingbox (-1,-1) rectangle (7,5);
	\clip (-1,-1) rectangle (7,5);
	\draw[thick] (-1,-1) rectangle (7,5);

\coordinate (P1) at (1, 0);
\coordinate (P2) at (2, 2);
\coordinate (P3) at (0, 2);
\coordinate (P4) at (4, 0);
\coordinate (P5) at (3, 4);
\coordinate (P6) at (6, 0.1);
\coordinate (S) at (3, 2);
\coordinate (S2) at (6, 0);

\draw[red] (P5) -- (S);
\draw[red] (P2) -- (S);
\draw[red] (P4) -- (S2);

\foreach \i/\r in {1/0.5,2/1,3/1.5,4/2,5/2.5} {
  
  \fill[gray] (P\i) circle ({0.15/\r});
}
\foreach \i/\r in {1/0.5,3/1.5,4/2,2/0.7} {
  
  \fill[blue] (P\i) circle ({0.15/\r});
  \draw[blue] (P\i) circle ({1.05*\r});
}
\fill[red] (P6) rectangle ++(0.2, -0.2);
\node at (0,4.3) {\small $i=1$};
\node at (2,1.45) {$s$};
\end{tikzpicture}
\begin{tikzpicture}[scale = 0.5485]
	\useasboundingbox (-1,-1) rectangle (7,5);
	\clip (-1,-1) rectangle (7,5);
	\draw[thick] (-1,-1) rectangle (7,5);

\coordinate (P1) at (1, 0);
\coordinate (P2) at (2, 2);
\coordinate (P3) at (0, 2);
\coordinate (P4) at (4, 0);
\coordinate (P5) at (3, 4);
\coordinate (P6) at (6, 0.1);
\coordinate (S2) at (6, 0);

\coordinate (S) at (3, 2);
\draw[red] (P5) -- (S);
\draw[red] (P2) -- (S);
\draw[red] (P4) -- (S2);
\draw[red] (P3) -- (P2);

\foreach \i/\r in {1/0.5,2/1,3/1.5,4/2,5/2.5} {
  
  \fill[gray] (P\i) circle ({0.15/\r});
}
\foreach \i/\r in {1/0.5,3/1.5,2/0.7} {
  
  \fill[blue] (P\i) circle ({0.15/\r});
  \draw[blue] (P\i) circle ({1.33*\r});
}
\fill[red] (P6) rectangle ++(0.2, -0.2);
\node at (0,4.3) {\small $i=2$};

\end{tikzpicture}
\begin{tikzpicture}[scale = 0.5485]
	\useasboundingbox (-1,-1) rectangle (7,5);
	\clip (-1,-1) rectangle (7,5);
	\draw[thick] (-1,-1) rectangle (7,5);

\coordinate (P1) at (1, 0);
\coordinate (P2) at (2, 2);
\coordinate (P3) at (0, 2);
\coordinate (P4) at (4, 0);
\coordinate (P5) at (3, 4);
\coordinate (P6) at (6, 0.1);
\coordinate (S4) at (2, 0);

\coordinate (S) at (3, 2);
\draw[red] (P5) -- (S);
\draw[red] (P2) -- (S);
\draw[red] (P4) -- (S2);
\draw[red] (P3) -- (P2);
\draw[red] (S4) -- (P2);
\draw[red] (S4) -- (P1);

\foreach \i/\r in {1/0.5,2/1,3/1.5,4/2,5/2.5} {
  
  \fill[gray] (P\i) circle ({0.15/\r});
}
\foreach \i/\r in {1/0.5,2/0.6} {
  
  \fill[blue] (P\i) circle ({0.15/\r});
  \draw[blue] (P\i) circle ({3.7*\r});
}
\fill[red] (P6) rectangle ++(0.2, -0.2);
\node at (0,4.3) {\small $i=3$};

\end{tikzpicture}
\begin{tikzpicture}[scale = 0.5485]
	\useasboundingbox (-1,-1) rectangle (7,5);
	\clip (-1,-1) rectangle (7,5);
	\draw[thick] (-1,-1) rectangle (7,5);
\coordinate (P1) at (1, 0);
\coordinate (P2) at (2, 2);
\coordinate (P3) at (0, 2);
\coordinate (P4) at (4, 0);
\coordinate (P5) at (3, 4);
\coordinate (P6) at (6, 0.1);
\coordinate (S4) at (2, 0);

\coordinate (S) at (3, 2);
\draw[red] (P5) -- (S);
\draw[red] (P2) -- (S);
\draw[red] (P4) -- (S2);
\draw[red] (P3) -- (P2);
\draw[red] (S4) -- (P2);
\draw[red] (S4) -- (P1);
\draw[red] (P1) -- (S2);
\foreach \i/\r in {1/0.5,2/1,3/1.5,4/2,5/2.5} {
  
  \fill[gray] (P\i) circle ({0.15/\r});
}
\foreach \i/\r in {1/0.4} {
  
  \fill[blue] (P\i) circle ({0.15/\r});
  \draw[blue] (P\i) circle (5.1);
}

\fill[red] (P6) rectangle ++(0.2, -0.2);
\node at (0,4.3) {\small $i=4$};

\end{tikzpicture}
\caption{Cost-distance algorithm example. Root in red, the sinks $S$ in black.
For simplicity the labeled vertices in global routing graph are visualized as  Euclidean balls.}
\label{fig:CD}
\end{figure}
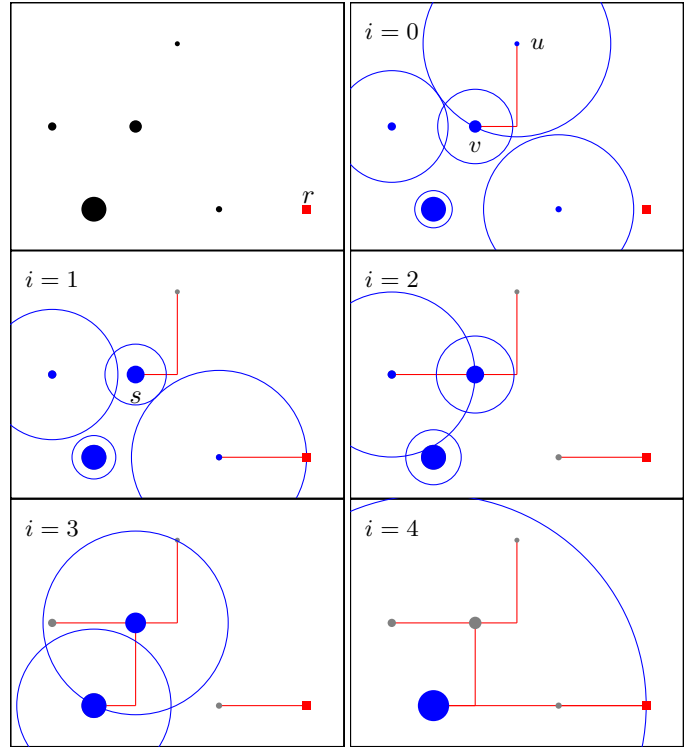

\subsection{An Example}
Figure \ref{fig:CD} visualizes the course  of the algorithm for a set of $5$ sinks, the black dots in the first picture.
The size of the dots represents their delay weight. We call the elements of $S_i$ in iteration $i$ the \textit{active} vertices (marked blue).
In iteration $i=0$ of the algorithm, $S_0 = S$, thus all sinks are active, and we start Dijkstra searches (areas inside blue circles) from all sinks using their individual cost functions. The larger the delay weight, the slower the circle grows.
The  terminal $u$ is the first one to find a connection (to $v$).

One of $\pi(u)$ and $\pi(v)$ is now chosen randomly with probability  proportional to its delay weight to place the Steiner vertex $s$ with weight $w(s) \coloneqq w(u)+w(v)$.
The random choice makes it more likely for the heavier weighted terminal to be reachable from $r$ without a detour.
Here, the more likely vertex $v$ was chosen.

The shortest $u$-$s$-$v$ path regarding $l_u$ (red) is added to $T$.
The terminals $u$ and $v$ are deactivated and $s$ is activated.
In the next iteration $i=1$, we started a new Dijkstra search from $s$ w.r.t. $l_s$.
In this iteration the shortest connection is a root connection, so we do not make a random choice for our Steiner vertex and just deactivate the sink.
We continue connecting vertices and decreasing the number of active vertices until all sinks are connected.

\subsection{Theoretical Properties}
We first analyze the running time of the algorithm.
\begin{theorem}
  \label{thm:running-time}
  The cost-distance algorithm can be implemented with a running time of $\mathcal{O}(t(n\log n + m))$.
\end{theorem}
\begin{proof}
  The total running time is dominated by the shortest path computations in Step~\ref{algline:dijkstra}.
  We start a Dijkstra search from each sink $u \in S_i$ simultaneously using the individual distance function $l_u$.
  
  Whenever we enter a vertex $v \in S_i \cup \{r_i\} \setminus \{u\}$, we add  the optimally balanced weighted node delay $b(u,v)$.
  If such a vertex gets permanently labeled, we have found a pair $u,v$  minimizing $L(u,v)$.
  After adding the connection, we start a Dijkstra search from the newly inserted Steiner vertex $s$.
  Dijkstra labels from other terminals remain valid in the next iteration.
  As we insert $|S|-1$ Steiner vertices, we have  $2|S|-1$ path searches, implementing shortest path searches with Fibonacci heaps
  \cite{DijkstraWithFibonacci:87} yields the claimed running time.
\end{proof}

We now prove the approximation guarantee of our algorithm.
Throughout, let  $\OPT$ denote the cost  according to (\ref{eqn:cost-distance-problem})  and (\ref{eqn:bif-delay-formula}) 
of an optimum solution $T^*$ to the original instance.
The challenge lies in bounding the incremental cost increase from each iteration.
Similar to \cite{meyerson-etal:08}, we show that there is a cheap matching of terminals w.r.t. weights $L$ (Lemma~\ref{lem:path-decomposition-weight}) which covers roughly the price increase of 
$\frac{1}{4}$ of the iterations (Theorem~\ref{thm:approximation-alg}).

In iteration $i$, consider the potential $$D_i \coloneqq \sum_{s\in S_i} w(s)\mathrm{delay}_{T^*}(r,s).$$
As new Steiner vertices  $s\in S_i$ are mapped to a position $\pi(s)= \pi(v')$, $ v' \in S\cup\{r\}$,
 $\mathrm{delay}_{T^*}(r,s) \coloneqq  \mathrm{delay}_{T^*}(r,v')$ is well-defined for all $s\in S_i$.

\begin{lemma}\label{expected_cost}
  In each iteration $i$, there is a cost-distance tree for the terminal set $S_i\cup \{r_i\}$
  with expected total cost at most $$C(T^*) +\mathbb{E}[D_i] \le \OPT.$$
\end{lemma}

\begin{proof}
  Suppose that two sinks $u, v \in S_i$ are matched in iteration $i$.
  They contribute a delay cost of $w(u)\cdot\mathrm{delay}_{T^*}(r,u) +w(v)\cdot\mathrm{delay}_{T^*}(r,v)$ to $D_i$ in iteration $i$. The expected contribution of the new Steiner vertex $s$ to $D_{i+1}$ is
    \begin{align*}
      &\sum_{x\in \{u,v\}}\frac{w(x)}{w(u)+w(v)}\cdot w(s)\cdot\mathrm{delay}_{T^*}(r,x) \\
      =\;& w(u)\cdot\mathrm{delay}_{T^*}(r,u) +w(v)\cdot\mathrm{delay}_{T^*}(r,u).
    \end{align*}
    Otherwise, if one of the merged nodes is the root, the new Steiner vertex does not contribute to $D_{i+1}$.
    This means that we have $\mathbb{E}[D_{i+1}] \leq \mathbb{E}[D_i]$.
    Thus,  $C(T^*)+\mathbb{E}[D_i] \leq \OPT$ by induction.
    As $T^*$ can be used as a solution for $S_i$, we have proven the lemma.
\end{proof}

A matching $M$  in a set  $X$ is called \textit{near-perfect} if it leaves at most
one element uncovered, i.e. $|X \setminus \cup_{e\in M}e| \le 1$.

\begin{lemma}
  Given a bifurcation compatible tree $T = (V,E)$ and a set of nodes $S \subseteq V$ , there exists a
  near-perfect matching of the nodes in $S$ such that the union of the unique paths in $T$ between matched nodes contains
  each edge and each vertex of the tree at most once.
  \label{lem:path-decomposition}
\end{lemma}
\begin{proof}
  Lemma 4.2 in \cite{meyerson-etal:08} proves that for an arbitrary
  tree there is such a matching using each edge at most once. If $T$
  is bifurcation compatible, this implies that each vertex is contained
  in at most one path.
\end{proof}

\begin{lemma}\label{tree_matching}
  In each iteration $i$, there exists a near-perfect matching $M$ of $S_i \cup \{r_i\}$
  with   $\mathbb{E}[L(M)] \le \OPT$.
  \label{lem:path-decomposition-weight}
\end{lemma}

\begin{proof}
  By Lemma~\ref{lem:path-decomposition},  there is a near-perfect matching $M$ in $S_i\cup \{r_i\}$ such that 
  the union of the unique paths in $T^*$ between matched nodes contains  each edge and each vertex at most once.
  
  For technical reasons, we extend  $w$ to $w(r_i)$ with an arbitrary constant.
  Given $\{u,v\} \in M$, we may assume w.l.o.g. $w(u) \leq w(v)$ or $v=r_i$. Then the cost of $M$ is
  \begin{align*}
    & \sum_{(u,v)\in M}L(u,v) \\
     \le  &  \sum_{(u,v)\in M} \Big(\dist_{T^*,c}(u,v) + w(u) \dist_{T^*,d}(u,v)  + b(u,v)\Big)\\
     \leq  & \:C(T^*) + \!  \sum_{(u,v)\in M}\! \Big(w(u)\dist_{T^*,d}(r,u) + w(v)\dist_{T^*,d}(r,v) \\
    & \qquad \qquad \qquad \quad \; \; + b(u,v)\Big)
      	     \leq  C(T^*) + D_i,
  \end{align*}
  where the first inequality holds due to
  Lemma~\ref{lem:path-decomposition}. The next inequality holds
  because we added twice the distance w.r.t $d$ from $r$ to the
  $u$-$v$ path in $T^*$ and $w(u)\le w(v)$. The last inequality holds because 
  $b(u,v)$ uses the optimum delay penalty  distribution, and the unique paths defined by the
  matching use each vertex exactly once. In contrast, $D_i$ uses a penalty distribution, where
  penalties can be counted multiple times.
  Observe that the $b(u,v)$ values are accounted for in $D_i$ either at the unique  least common ancestor  (w.r.t. the root $r$) or along the vertices of an $r$-$u$-path.

  By Lemma \ref{expected_cost}, $\mathbb{E}[C(T^*) + D_i] \le \OPT$.
\end{proof}

\begin{lemma}\label{transfer}
  When the algorithm merges  $u\in S_i$ and $v \in S_i\cup\{r_i\}$  in iteration $i$,
  $L(u,v)$ reflects at least half of the expected objective cost increase in this iteration.
\end{lemma}

\begin{proof}
  $L(u,v)$ reflects the correct connection cost and bifurcation penalty $b(u,v)$.
  Let $P$ be the underlying path reflecting the  embedding  into $G$.
  If $v\neq r_i$, the vertex $x\in \{u,v\}$ 
  is chosen with a probability  of $w(x)/(w(u)+w(v))$.
  Thus, the expected delay cost  w.r.t. $d$ is
  \begin{align*}
      \frac{2 w(u)w(v)}{w(u)+w(v)}\dist_{P,d}(u,v) 
    \!\leq\!  2 \min\{w(u),w(v)\}\!\dist_{P,d}(u,v)
  \end{align*}
  and 
  \begin{align*}
    &2 \min\{w(u),w(v)\}\dist_{P,d}(u,v) + \dist_{P,c}(u,v) + b(u,v) \\
   \leq \:& 2L(u,v)
 \end{align*}
  If $v = r_i$, the (deterministic) cost  is $L(u,v)$.
\end{proof}

\begin{theorem}
  The cost-distance algorithm computes a Steiner tree  with expected objective value at most  $\mathcal{O}(\log t)\cdot \OPT$.
  \label{thm:approximation-alg}
\end{theorem}
\begin{proof}
  Consider a minimum weight near-perfect matching $M$ in ($S\cup \{r\}$).
  Order the matching edges $M=\{\{u_1,v_1\}, \dots, \{u_{\lfloor t/2\rfloor},v_{\lfloor t/2\rfloor}\}\}$ in non-decreasing order w.r.t. $L$.
      
  For the first $\lfloor t/4 \rfloor$ iterations, the $i$-th (shortest) path from our algorithm will have costs at most   $L(u_{2i-1},v_{2i-1})$,
  because each path can intersect at most two matching edges.
  Thus, we paid at most half of the cost of an optimal near perfect matching when adding up the $\lfloor t/4 \rfloor$ path costs $L(u,v)$.
  By Lemmas \ref{tree_matching} and \ref{transfer}, this implies a cost of at most $\OPT$ for reducing the size of our instance by a factor (almost) $3/4$.

  Reducing the size repeatedly, we pay $\mathcal{O}(\log_{4/3}t)$ times $\OPT$ overall.
\end{proof}

\section{Practical Enhancements}
\label{sec:pracical-enhancements}
In this section, we present multiple practical enhancements.

\subsection{Discounting Existing Tree Components} \label{edgecost_reduction}
The variant from Section~\ref{sec:approx-algorithm} inserts Steiner vertices only at terminal positions, which  leads to  unnecessary detours.
To avoid this, we set $c(e) = 0$ for edges in the connected components of $T$ where the current path search starts and ends, and allow reusing them.
We add separate labels for each possible end component a path search reaches so that the zero connection cost edges can only be used to reach that particular end component.
This implicitly places Steiner vertices at the points where the path leaves or enters the connected components. We still restart the next path search from one of the randomly chosen terminal positions.
This change significantly improves connection costs and preserves the theoretical guarantee if $d_{bif}=0$.

\subsection{Two-level Heap Structure}
As global routing graphs usually satisfy $m \in \mathcal{O}(n)$,
we use binary heaps.  We use
one heap for each sink in $S_i$ and a top-level heap that stores the minimum
key from each sink heap.
We operate  with a single sink heap until the minimum label in the  top-level heap is exceeded,
and we extract the next minimum globally.

\subsection{Goal-Oriented Path Searches}
To speed up the path searches, we can use goal-oriented path searches, a.k.a. $A^*$ searches \cite{Astar-algorith:68}.
This reduces the number of labeling steps. 
In our implementation, the connection/congestion costs between two vertices are lower bounded using landmark based future costs \cite{goldberg:05} and the  delays are bounded based on $L_1$-distance and the fastest layer and wire type combination for that distance.

In a path search from a sink $u$,  we minimize the sum of these lower bounds over all targets in $S_i\cup\{r_i\}\setminus\{u\}$.
When setting the connection cost of edges in a component to zero as described in \ref{edgecost_reduction},
we need to make sure that the future costs stay feasible.
We do this by omitting the connection future costs until every vertex of the starting component is permanently labeled and subtracting an upper bound on the connection future cost for the component containing
a target $v \in S_i\cup\{r_i\}\setminus \{u\}$.

\subsection{Better Embedding of Steiner Vertices}
\label{sec:steiner-vertex-positions}
We can improve the position of the Steiner vertices $s$.
After adding a path $P$ between terminals $u$ and $v$, we choose the position $\pi(s)$ on $P$ that minimizes
$$c(E(Q)) + (w(u) + w(v))d(E(Q)) + \sum_{y\in \{u,v\}}w(y)d(E(P_{[y, s]})) ,$$
where $Q$ is a cheapest $s$-$r$-path w.r.t. $c + (w(u) + w(v))d$.
This corresponds to extending $P$ to an optimum cost-distance tree for the sink set $S = \{u,v\}$.
Finding the best Steiner vertex $s$ requires an extra path search to compute $Q$. Instead, we use future costs to estimate it.
Though, we lose the theoretical guarantees, this modification  improves the quality in practice.

If the new Steiner vertex $s$ was  already permanently labeled w.r.t. $l_v$,
we can use this path to  get some further saving.

\subsection{Encourage Root Connections}
Lastly, connecting a terminal $u$ to the root $r_i$ early in the algorithm is very expensive because the expected node delay cost $b(u,r_i)$  reflects all delay weights.
However, connecting to the root immediately reduces the cost of an optimal tree for $S_{i+1}$,
saving a cost of at least $\eta \cdot d_{bif}\cdot w(u)$ in future iterations.
We subtract these cost savings from the node delay costs for root connections to encourage merging with $r_i$.

\section{Experimental Results}
\label{sec:experimental-results}
We integrated our cost-distance Steiner tree function as a Steiner
tree oracle of a timing-constrained global routing
 algorithm similar to  \cite{Held-etal:18}. Lagrangean relaxation is used to relax
global timing and routing constraints. The cost-distance Steiner tree
problem  (\ref{eqn:cost-distance-problem}) arises as the Lagrangean subproblem.

\subsection{Algorithms for Comparison}
We compare the new cost-distance Steiner tree algorithm  with three alternatives.
Each first computes a Steiner topology in the plane, considering total length instead of congestion cost.
Then, this tree is embedded optimally into the global routing graph minimizing the cost-distance objective (\ref{eqn:cost-distance-problem}) using a Dijkstra-style embedding as described in \cite{Held-etal:18}.
The embedding greatly determines the signal speed by choosing layers and wire types.

The first routine just computes a short $L_1$ Steiner tree and embeds it optimally into the global routing graph. 

The second is based on shallow-light Steiner trees \cite{rotter-held:13,SALT:19}.
These algorithms start from an approximately minimum-length
tree.  During a DFS traversal, sinks are reconnected to the root 
whenever they violate a given delay/distance bound by more than a
factor $(1+\varepsilon)$ for an $\varepsilon >0$. In a reverse DFS traversal, deleted edges
may be re-activated to connect former predecessors if that saves cost.
In \cite{rotter-held:13} bifurcation penalties were
also integrated for $\eta=0.5$.
We (re-)distribute them according to our more flexible model in the initial tree and during the reverse DFS
traversal.

Another very popular method is the Prim-Dijkstra algorithm
\cite{PrimDijkstra:95,PrimDijkstra:18}. Here, sinks are iteratively added into the
root-component.  A sink $s$ and an edge $e$ in the root component are
chosen to insert a new Steiner vertex into $e$ connecting $s$ such
that a weighted sum of total length and path length to $s$ is
minimized.  Bifurcation penalties were incorporated into this method 
 for the special case $\eta = 0.5$ \cite{Bartoschek-etal:10}. We can distribute the delay penalty to the two
branches, when selecting the edge of the root component.

Both methods use the globally optimized delay budgets that arise
from the resource sharing algorithm \cite{Held-etal:18}.

The three  methods are well established in academic and commercial design
tools and form a solid foundation to compare against. They are
denoted as $L_1$ ($L_1$-shortest), SL (shallow-light), and PD (Prim-Dijkstra) in the following
tables. Our new cost-distance algorithm with the practical enhancements is denoted CD.

We made tests with and without bifurcation delay penalties.

\subsection{Results on Individual Instances}

\begin{table}[tb]
  \renewcommand{\arraystretch}{1.1}
	\caption{Average cost increase compared to minimum, with $d_{bif} = 0$.}
	\label{table_cc_no_nodedelay}
	\centering
	\begin{tabular}{ r | r | r | r | r | r }
		$|S|$&  \#instances & $L_1$ & SL & PD & CD \\ \hline 
	   3-5 & 527784 & \textbf{0.40}\% & 0.41\% & 0.44\%& 0.67\%  \\
	   6-14 & 156053 & 1.62\% & 1.58\% & \textbf{1.54}\%& 1.57\%  \\
	   15-29 & 57594 & 4.70\% & 3.73\% & 3.39\%& \textbf{2.24}\%  \\
	   $\geq30$ & 84033 & 7.09\% & 5.38\% & 4.48\%& \textbf{1.73}\%  \\
	   \hline
	   all & 825464 & 1.61\% & 1.37\% & 1.27\% &\textbf{1.06}\% 
	\end{tabular}
\end{table}
\begin{table}[tb]
	\renewcommand{\arraystretch}{1.1}
	\caption{Average cost increase compared to minimum, with $d_{bif} > 0$.}
	\label{table_cc_with_nodedelay}
	\centering
	\begin{tabular}{ r | r | r | r | r | r }
		$|S|$&  \#instances & $L_1$ & SL & PD & CD \\ \hline 
	   3-5 & 741967 & 0.96\% & 0.95\% & 1.20\%& \textbf{0.65}\%  \\
	   6-14 & 222419 & 4.52\% & 4.32\% & 3.28\%& \textbf{1.89}\%  \\
	   15-29 & 84034 & 11.37\% & 9.76\% & 4.82\%& \textbf{3.03}\%  \\
	   $\geq30$ & 127861 & 23.73\% & 13.39\% & 7.07\% & \textbf{3.30}\%  \\ \hline
	   all & 1176281 & 4.85\% & 3.57\% & 2.49\% & \textbf{1.34}\% 
	\end{tabular}
\end{table}

First, we present an apples-to-apples comparisons
on identical  cost-distance Steiner tree instances as they were generated during timing-constrained global routing on several 5nm designs,
where we set $d_{bif} =  0$ in the first set of tests shown in Table \ref{table_cc_no_nodedelay}.
For each instance, we measured the relative increase in the objective (\ref{eqn:cost-distance-problem}) between each algorithm and the best of the 4 algorithms.
On the small instances with 3-5 sinks, the $L_1$  heuristic gives the best results, though results don't vary much between the 4 algorithms.
Note that the embedding into the 3D global routing graph is optimum w.r.t. (\ref{eqn:cost-distance-problem}) for $L_1$, SL, and PD.
On larger instances with more than 15 sinks, the cost-distance algorithm  dominates the other methods.

Table \ref{table_cc_with_nodedelay} shows the same results with bifurcation penalties, here the cost-distance algorithm dominates the other methods throughout.

\subsection{Global Routing Results}
Next, we compare the impact of the different algorithms on
timing-constrained global routing results.  Congestion is measured using the
$ACE$  \cite{Ace}.  $ACE(x)$ is the
average congestion of the $x\%$ most critical global routing edges.
We then use $ACE4 \coloneqq \frac{1}{4}(ACE(.5)+ACE(1)+ACE(2)+ACE(5))$. Note
that the most critical $0.5$-percentile is implicitly contributing to all
four numbers.
An  $ACE4$ of $93\%$ is usually considered routable, but detailed routing will
introduce significant detours and delays if it is above $90\%$.

\begin{table}[tb]
  \caption{Instance parameters (all are in 5nm technology)}
  \label{tab:instance-sizes}
  \centering
  \begin{tabular}{ rrrr rrrr}
    \toprule 
    Chip & \# nets & \# layers \\
    \midrule 
    c1 & 49\,734 & 8  \\
    c2 & 66\,500 & 9  \\
    c3 & 286\,619 &7 \\
    c4 & 305\,094 & 15\\
    \bottomrule
  \end{tabular}
\quad
\begin{tabular}{ rrrr rrrr}
    \toprule     
    Chip & \#  nets & \# layers  \\  
    \midrule 
    c5 & 420\,131 & 9 \\
    c6 & 590\,060 & 9 \\
    c7 & 650\,127 & 15\\
    c8 & 941\,271 & 15 \\
    \bottomrule
  \end{tabular}
\end{table}
Table~\ref{tab:instance-sizes} gives an overview of the instances. They  are a mix of industrial microprocessor and ASIC units.
Again, we start with experiments without bifurcation penalties ($d_{bif} = 0$) in Table~\ref{tab:results-no-nodedelay}.
For each instance and each Steiner method, as well as for the sum/average of all instances it contains a line, and we report the worst slack (WS), total negative slack (TNS),
the $ACE4$, the wire length (WL), the via count (Vias), and the wall times on an AMD EPYC 9684X processor using 16 threads.
Best results are marked in green.
For most instances our new cost-distance Steiner trees yield the best timing (WS \& TNS), the best $ACE4$  and the best via count. The topology in the other routines will often force embedding path searches through congested areas, increasing their via count.
However,  cost-distance trees come  with a higher wire length.
The running time of all methods are comparable.
\begin{table}[tb]
  \caption{Timing-constrained global routing results with $d_{bif} = 0$.}
  \label{tab:results-no-nodedelay}  
  \scriptsize
  \addtolength{\tabcolsep}{-0.2em}
\begin{tabular}{ llHrHrHrHrHrHrH }
  \toprule
Chip
  & {Run}
  & {Nets}
  & WS &
  & TNS&
  & ACE4 &
  & WL  &
  & Vias &
  & Walltime &
  \\
  & 
  & {Nets}
  & [ps] &
  & [ps] &
  & [\%] &
  & [m] &
  & [\#] &
  & [h:m:s] &
  \\

  \midrule


     {c1}
   & {$L_1$}
   & \rmUkn {49\,734}
   & \rmUkn {-49}
   & {}
   & \rmUkn {-1\,633}
   & {}
   & \rmUkn {91.21}
   & {}
   & \rmBtr {0.6016\,\unit{ }}
   & {}
   & \rmUkn {561\,140}
   & {}
   & \rmUkn {\formattime{00}{03}{24}}
   & {}
  \\

     {}
   & {SL}
   & \rmUkn {49\,734}
   & \rmUkn {-49}
   & \rmEql {0}
   & \rmUkn {-1\,544}
   & \rmBtr {89}
   & \rmUkn {91.32}
   & \rmEql {0.11}
   & \rmUkn {0.6073\,\unit{ }}
   & \rmEql {+0.95}
   & \rmUkn {560\,678}
   & \rmEql {-0.08}
   & \rmUkn {\formattime{00}{03}{23}}
   & \rmEql {-0.49}
  \\

     {}
   & {PD}
   & \rmUkn {49\,734}
   & \rmUkn {-49}
   & \rmEql {0}
   & \rmUkn {-1\,370}
   & \rmBtr {263}
   & \rmUkn {91.50}
   & \rmEql {0.29}
   & \rmUkn {0.6104\,\unit{ }}
   & \rmEql {+1.46}
   & \rmUkn {558\,133}
   & \rmEql {-0.54}
   & \rmUkn {\formattime{00}{05}{29}}
   & \rmWrs {+61.27}
  \\

     {}
   & {CD}
   & \rmUkn {49\,734}
   & \rmUkn {-49}
   & \rmEql {0}
   & \rmBtr {-1\,340}
   & \rmBtr {294}
   & \rmBtr {89.36}
   & \rmBtr {-1.85}
   & \rmUkn {0.6152\,\unit{ }}
   & \rmEql {+2.26}
   & \rmBtr {547\,240}
   & \rmBtr {-2.48}
   & \rmUkn {\formattime{00}{03}{30}}
   & \rmWrs {+2.94}
  \\

  \midrule

     {c2}
   & {$L_1$}
   & \rmUkn {66\,500}
   & \rmUkn {-82}
   & {}
   & \rmUkn {-19\,774}
   & {}
   & \rmUkn {86.36}
   & {}
   & \rmBtr {1.0730\,\unit{ }}
   & {}
   & \rmUkn {864\,387}
   & {}
   & \rmUkn {\formattime{00}{04}{30}}
   & {}
  \\

     {}
   & {SL}
   & \rmUkn {66\,500}
   & \rmUkn {-80}
   & \rmBtr {2}
   & \rmUkn {-19\,659}
   & \rmBtr {115}
   & \rmUkn {85.75}
   & \rmBtr {-0.61}
   & \rmUkn {1.0838\,\unit{ }}
   & \rmEql {+1.01}
   & \rmUkn {863\,571}
   & \rmEql {-0.09}
   & \rmUkn {\formattime{00}{04}{00}}
   & \rmBtr {-11.11}
  \\

     {}
   & {PD}
   & \rmUkn {66\,500}
   & \rmUkn {-80}
   & \rmBtr {2}
   & \rmUkn {-19\,558}
   & \rmBtr {216}
   & \rmUkn {85.81}
   & \rmBtr {-0.55}
   & \rmUkn {1.0845\,\unit{ }}
   & \rmEql {+1.07}
   & \rmUkn {862\,924}
   & \rmEql {-0.17}
   & \rmUkn {\formattime{00}{03}{34}}
   & \rmBtr {-20.74}
  \\

     {}
   & {CD}
   & \rmUkn {66\,500}
   & \rmUkn {-80}
   & \rmBtr {2}
   & \rmBtr {-19\,512}
   & \rmBtr {262}
   & \rmBtr {85.56}
   & \rmBtr {-0.80}
   & \rmUkn {1.1033\,\unit{ }}
   & \rmEql {+2.82}
   & \rmBtr {854\,387}
   & \rmBtr {-1.16}
   & \rmUkn {\formattime{00}{03}{52}}
   & \rmBtr {-14.07}
  \\

  \midrule

     {c3}
   & {$L_1$}
   & \rmUkn {286\,619}
   & \rmUkn {-220}
   & {}
   & \rmUkn {-105\,137}
   & {}
   & \rmUkn {92.07}
   & {}
   & \rmBtr {2.0070\,\unit{ }}
   & {}
   & \rmUkn {2\,336\,123}
   & {}
   & \rmUkn {\formattime{00}{17}{33}}
   & {}
  \\

     {}
   & {SL}
   & \rmUkn {286\,619}
   & \rmUkn {-220}
   & \rmEql {0}
   & \rmUkn {-82\,142}
   & \rmBtr {22\,995}
   & \rmUkn {92.28}
   & \rmEql {0.21}
   & \rmUkn {2.0305\,\unit{ }}
   & \rmEql {+1.17}
   & \rmUkn {2\,327\,893}
   & \rmEql {-0.35}
   & \rmUkn {\formattime{00}{18}{28}}
   & \rmWrs {+5.22}
  \\

     {}
   & {PD}
   & \rmUkn {286\,619}
   & \rmUkn {-220}
   & \rmEql {0}
   & \rmUkn {-67\,763}
   & \rmBtr {37\,374}
   & \rmUkn {92.39}
   & \rmEql {0.32}
   & \rmUkn {2.0375\,\unit{ }}
   & \rmEql {+1.52}
   & \rmUkn {2\,324\,940}
   & \rmEql {-0.48}
   & \rmUkn {\formattime{00}{20}{23}}
   & \rmWrs {+16.14}
  \\

     {}
   & {CD}
   & \rmUkn {286\,619}
   & \rmUkn {-220}
   & \rmEql {-0}
   & \rmBtr {-67\,643}
   & \rmBtr {37\,494}
   & \rmBtr {91.54}
   & \rmBtr {-0.53}
   & \rmUkn {2.0751\,\unit{ }}
   & \rmEql {+3.39}
   & \rmBtr {2\,313\,590}
   & \rmEql {-0.96}
   & \rmUkn {\formattime{00}{16}{40}}
   & \rmBtr {-5.03}
  \\

  \midrule

     {c4}
   & {$L_1$}
   & \rmUkn {305\,094}
   & \rmBtr {-67}
   & {}
   & \rmUkn {-31\,907}
   & {}
   & \rmUkn {90.23}
   & {}
   & \rmBtr {6.5801\,\unit{ }}
   & {}
   & \rmUkn {4\,095\,449}
   & {}
   & \rmUkn {\formattime{00}{36}{15}}
   & {}
  \\

     {}
   & {SL}
   & \rmUkn {305\,094}
   & \rmUkn {-68}
   & \rmEql {-0}
   & \rmUkn {-31\,805}
   & \rmBtr {102}
   & \rmUkn {90.12}
   & \rmEql {-0.11}
   & \rmUkn {6.6114\,\unit{ }}
   & \rmEql {+0.48}
   & \rmUkn {4\,083\,044}
   & \rmEql {-0.30}
   & \rmUkn {\formattime{00}{50}{07}}
   & \rmWrs {+38.25}
  \\

     {}
   & {PD}
   & \rmUkn {305\,094}
   & \rmUkn {-71}
   & \rmWrs {-4}
   & \rmUkn {-31\,769}
   & \rmBtr {138}
   & \rmUkn {90.21}
   & \rmEql {-0.02}
   & \rmUkn {6.6270\,\unit{ }}
   & \rmEql {+0.71}
   & \rmUkn {4\,067\,446}
   & \rmEql {-0.68}
   & \rmUkn {\formattime{00}{49}{31}}
   & \rmWrs {+36.60}
  \\

     {}
   & {CD}
   & \rmUkn {305\,094}
   & \rmUkn {-68}
   & \rmEql {-0}
   & \rmBtr {-31\,682}
   & \rmBtr {225}
   & \rmBtr {89.77}
   & \rmEql {-0.46}
   & \rmUkn {6.8127\,\unit{ }}
   & \rmEql {+3.53}
   & \rmBtr {4\,005\,655}
   & \rmBtr {-2.19}
   & \rmUkn {\formattime{00}{43}{32}}
   & \rmWrs {+20.09}
  \\

  \midrule

     {c5}
   & {$L_1$}
   & \rmUkn {420\,131}
   & \rmUkn {-276}
   & {}
   & \rmUkn {-946\,089}
   & {}
   & \rmUkn {86.26}
   & {}
   & \rmBtr {4.4107\,\unit{ }}
   & {}
   & \rmUkn {4\,640\,578}
   & {}
   & \rmUkn {\formattime{00}{21}{29}}
   & {}
  \\

     {}
   & {SL}
   & \rmUkn {420\,131}
   & \rmUkn {-286}
   & \rmWrs {-10}
   & \rmUkn {-869\,512}
   & \rmBtr {76\,577}
   & \rmUkn {86.15}
   & \rmEql {-0.11}
   & \rmUkn {4.4633\,\unit{ }}
   & \rmEql {+1.19}
   & \rmUkn {4\,612\,457}
   & \rmEql {-0.61}
   & \rmUkn {\formattime{00}{26}{23}}
   & \rmWrs {+22.81}
  \\

     {}
   & {PD}
   & \rmUkn {420\,131}
   & \rmBtr {-264}
   & \rmBtr {12}
   & \rmUkn {-803\,343}
   & \rmBtr {142\,746}
   & \rmUkn {85.94}
   & \rmEql {-0.32}
   & \rmUkn {4.4687\,\unit{ }}
   & \rmEql {+1.31}
   & \rmUkn {4\,580\,028}
   & \rmBtr {-1.30}
   & \rmUkn {\formattime{00}{26}{58}}
   & \rmWrs {+25.52}
  \\

     {}
   & {CD}
   & \rmUkn {420\,131}
   & \rmBtr {-264}
   & \rmBtr {12}
   & \rmBtr {-776\,373}
   & \rmBtr {169\,717}
   & \rmBtr {85.27}
   & \rmBtr {-0.99}
   & \rmUkn {4.6067\,\unit{ }}
   & \rmEql {+4.44}
   & \rmBtr {4\,522\,756}
   & \rmBtr {-2.54}
   & \rmUkn {\formattime{00}{21}{25}}
   & \rmEql {-0.31}
  \\

  \midrule

     {c6}
   & {$L_1$}
   & \rmUkn {590\,060}
   & \rmUkn {-331}
   & {}
   & \rmUkn {-536\,208}
   & {}
   & \rmUkn {85.96}
   & {}
   & \rmBtr {3.6925\,\unit{ }}
   & {}
   & \rmUkn {4\,699\,131}
   & {}
   & \rmUkn {\formattime{00}{35}{08}}
   & {}
  \\

     {}
   & {SL}
   & \rmUkn {590\,060}
   & \rmUkn {-314}
   & \rmBtr {17}
   & \rmUkn {-525\,484}
   & \rmBtr {10\,724}
   & \rmUkn {86.03}
   & \rmEql {0.07}
   & \rmUkn {3.7549\,\unit{ }}
   & \rmEql {+1.69}
   & \rmUkn {4\,678\,851}
   & \rmEql {-0.43}
   & \rmUkn {\formattime{00}{43}{27}}
   & \rmWrs {+23.67}
  \\

     {}
   & {PD}
   & \rmUkn {590\,060}
   & \rmUkn {-311}
   & \rmBtr {20}
   & \rmUkn {-509\,299}
   & \rmBtr {26\,909}
   & \rmUkn {86.02}
   & \rmEql {0.06}
   & \rmUkn {3.7537\,\unit{ }}
   & \rmEql {+1.66}
   & \rmUkn {4\,659\,440}
   & \rmEql {-0.84}
   & \rmUkn {\formattime{00}{39}{15}}
   & \rmWrs {+11.72}
  \\

     {}
   & {CD}
   & \rmUkn {590\,060}
   & \rmBtr {-308}
   & \rmBtr {23}
   & \rmBtr {-507\,841}
   & \rmBtr {28\,367}
   & \rmBtr {85.02}
   & \rmBtr {-0.94}
   & \rmUkn {3.8492\,\unit{ }}
   & \rmEql {+4.24}
   & \rmBtr {4\,613\,520}
   & \rmBtr {-1.82}
   & \rmUkn {\formattime{00}{37}{18}}
   & \rmWrs {+6.17}
  \\

  \midrule

     {c7}
   & {$L_1$}
   & \rmUkn {650\,127}
   & \rmBtr {-86}
   & {}
   & \rmUkn {-57\,267}
   & {}
   & \rmUkn {87.72}
   & {}
   & \rmBtr {13.1894\,\unit{ }}
   & {}
   & \rmUkn {7\,817\,717}
   & {}
   & \rmUkn {\formattime{01}{22}{23}}
   & {}
  \\

     {}
   & {SL}
   & \rmUkn {650\,127}
   & \rmUkn {-87}
   & \rmEql {-1}
   & \rmUkn {-57\,903}
   & \rmWrs {-636}
   & \rmUkn {87.69}
   & \rmEql {-0.03}
   & \rmUkn {13.2282\,\unit{ }}
   & \rmEql {+0.29}
   & \rmUkn {7\,800\,550}
   & \rmEql {-0.22}
   & \rmUkn {\formattime{01}{19}{05}}
   & \rmBtr {-4.01}
  \\

     {}
   & {PD}
   & \rmUkn {650\,127}
   & \rmUkn {-90}
   & \rmWrs {-4}
   & \rmUkn {-59\,139}
   & \rmWrs {-1872}
   & \rmUkn {87.61}
   & \rmEql {-0.11}
   & \rmUkn {13.2473\,\unit{ }}
   & \rmEql {+0.44}
   & \rmUkn {7\,790\,466}
   & \rmEql {-0.35}
   & \rmUkn {\formattime{01}{18}{02}}
   & \rmBtr {-5.28}
  \\

     {}
   & {CD}
   & \rmUkn {650\,127}
   & \rmUkn {-88}
   & \rmWrs {-2}
   & \rmBtr {-54\,659}
   & \rmBtr {2\,608}
   & \rmBtr {87.23}
   & \rmEql {-0.49}
   & \rmUkn {13.4631\,\unit{ }}
   & \rmEql {+2.08}
   & \rmBtr {7\,721\,547}
   & \rmBtr {-1.23}
   & \rmUkn {\formattime{01}{38}{58}}
   & \rmWrs {+20.13}
  \\

  \midrule

     {c8}
   & {$L_1$}
   & \rmUkn {941\,271}
   & \rmUkn {-102}
   & {}
   & \rmUkn {-174\,288}
   & {}
   & \rmUkn {91.14}
   & {}
   & \rmBtr {13.5765\,\unit{ }}
   & {}
   & \rmUkn {12\,559\,176}
   & {}
   & \rmUkn {\formattime{02}{00}{44}}
   & {}
  \\

     {}
   & {SL}
   & \rmUkn {941\,271}
   & \rmBtr {-99}
   & \rmBtr {3}
   & \rmUkn {-175\,643}
   & \rmWrs {-1355}
   & \rmUkn {91.05}
   & \rmEql {-0.09}
   & \rmUkn {13.6386\,\unit{ }}
   & \rmEql {+0.46}
   & \rmUkn {12\,519\,244}
   & \rmEql {-0.32}
   & \rmUkn {\formattime{02}{35}{41}}
   & \rmWrs {+28.95}
  \\

     {}
   & {PD}
   & \rmUkn {941\,271}
   & \rmUkn {-104}
   & \rmWrs {-1}
   & \rmBtr {-171\,323}
   & \rmBtr {2\,966}
   & \rmUkn {90.99}
   & \rmEql {-0.15}
   & \rmUkn {13.6962\,\unit{ }}
   & \rmEql {+0.88}
   & \rmUkn {12\,501\,090}
   & \rmEql {-0.46}
   & \rmUkn {\formattime{02}{35}{44}}
   & \rmWrs {+28.99}
  \\

     {}
   & {CD}
   & \rmUkn {941\,271}
   & \rmUkn {-105}
   & \rmWrs {-2}
   & \rmUkn {-174\,831}
   & \rmWrs {-542}
   & \rmBtr {90.78}
   & \rmEql {-0.36}
   & \rmUkn {14.3538\,\unit{ }}
   & \rmWrs {+5.73}
   & \rmBtr {12\,330\,992}
   & \rmBtr {-1.82}
   & \rmUkn {\formattime{02}{33}{09}}
   & \rmWrs {+26.85}
  \\



  \midrule
     {all}
   & {$L_1$}
   & {}
   & \rmUkn {-1\,214}
   & {}
   & \rmUkn {-1\,872\,304}
   & {}
   & \rmUkn {88.87}
   & {}
   & \rmBtr {45.1308\,\unit{ }}
   & {}
   & \rmUkn {37\,573\,701}
   & {}
   & \rmUkn {\formattime{05}{21}{26}}
   & {}
  \\

     {}
   & {SL}
   & {}
   & \rmUkn {-1\,204}
   & {}
   & \rmUkn {-1\,763\,692}
   & {}
   & \rmUkn {88.80}
   & \rmEql {-0.07}
   & \rmUkn {45.4180\,\unit{ }}
   & \rmEql {+0.64}
   & \rmUkn {37\,446\,288}
   & \rmEql {-0.34}
   & \rmUkn {\formattime{06}{20}{34}}
   & \rmWrs {+18.40}
  \\

     {}
   & {PD}
   & {}
   & \rmUkn {-1\,188}
   & {}
   & \rmUkn {-1\,663\,564}
   & {}
   & \rmUkn {88.81}
   & \rmEql {-0.06}
   & \rmUkn {45.5253\,\unit{ }}
   & \rmEql {+0.87}
   & \rmUkn {37\,344\,467}
   & \rmEql {-0.61}
   & \rmUkn {\formattime{06}{18}{56}}
   & \rmWrs {+17.89}
  \\

     {}
   & {CD}
   & {}
   & \rmBtr {-1\,181}
   & {}
   & \rmBtr {-1\,633\,880}
   & {}
   & \rmBtr {88.07}
   & \rmBtr {-0.80}
   & \rmUkn {46.8791\,\unit{ }}
   & \rmEql {+3.87}
   & \rmBtr {36\,909\,687}
   & \rmBtr {-1.77}
   & \rmUkn {\formattime{06}{18}{24}}
   & \rmWrs {+17.72}
  \\


  \bottomrule
\end{tabular}

\end{table}

Table~\ref{tab:results-with-nodedelay} shows results with bifurcation penalties.
Note that the bifurcation penalties increase the delays and thus decrease slacks.
The wire length and via count decrease compared to $d_{bif}=0$ because the delay prices now weigh stronger compared to congestion.
Again the new CD methods outperforms the others in terms of WS, TNS and via count, but
has a slightly larger wire length.
\begin{table}[tb]
  \caption{Timing-constrained global routing results with $d_{bif} > 0$.}
  \label{tab:results-with-nodedelay}
  \scriptsize
  \addtolength{\tabcolsep}{-0.2em}
\begin{tabular}{ llHrHrHrHrHrHrH }
  \toprule
Chip
  & {Run}
  & {Nets}
  & WS &
  & TNS&
  & ACE4 &
  & WL &
  & Vias &
  & Walltime &
  \\
  & 
  & {Nets}
  & [ps] &
  & [ps] &
  & [\%] &
  & [m] &
  & [\#] &
  & [h:m:s] &
  \\

  \midrule


     {c1}
   & {$L_1$}
   & \rmUkn {49\,734}
   & \rmUkn {-84}
   & {}
   & \rmUkn {-9\,837}
   & {}
   & \rmUkn {92.15}
   & {}
   & \rmBtr {0.5880\,\unit{ }}
   & {}
   & \rmUkn {536\,492}
   & {}
   & \rmUkn {\formattime{00}{03}{14}}
   & {}
  \\

     {}
   & {SL}
   & \rmUkn {49734}
   & \rmUkn {-64}
   & \rmBtr {20}
   & \rmUkn {-6\,325}
   & \rmBtr {3512}
   & \rmUkn {92.33}
   & \rmEql {0.18}
   & \rmUkn {0.6039\,\unit{ }}
   & \rmEql {+2.70}
   & \rmUkn {536\,398}
   & \rmEql {-0.02}
   & \rmUkn {\formattime{00}{03}{11}}
   & \rmBtr {-1.55}
  \\

     {}
   & {PD}
   & \rmUkn {49734}
   & \rmUkn {-64}
   & \rmBtr {20}
   & \rmBtr {-5\,937}
   & \rmBtr {3900}
   & \rmUkn {92.16}
   & \rmEql {0.01}
   & \rmUkn {0.6051\,\unit{ }}
   & \rmEql {+2.91}
   & \rmUkn {530\,699}
   & \rmBtr {-1.08}
   & \rmUkn {\formattime{00}{03}{34}}
   & \rmWrs {+10.31}
  \\

     {}
   & {CD}
   & \rmUkn {49734}
   & \rmBtr {-61}
   & \rmBtr {23}
   & \rmUkn {-6\,333}
   & \rmBtr {3504}
   & \rmBtr {91.93}
   & \rmEql {-0.22}
   & \rmUkn {0.6109\,\unit{ }}
   & \rmEql {+3.89}
   & \rmBtr {526\,427}
   & \rmBtr {-1.88}
   & \rmUkn {\formattime{00}{03}{09}}
   & \rmBtr {-2.58}
  \\

  \midrule

     {c2}
   & {$L_1$}
   & \rmUkn {66500}
   & \rmUkn {-105}
   & {}
   & \rmUkn {-28\,844}
   & {}
   & \rmUkn {87.76}
   & {}
   & \rmBtr {1.0728\,\unit{ }}
   & {}
   & \rmUkn {759\,747}
   & {}
   & \rmUkn {\formattime{00}{03}{41}}
   & {}
  \\

     {}
   & {SL}
   & \rmUkn {66500}
   & \rmUkn {-93}
   & \rmBtr {13}
   & \rmUkn {-22\,320}
   & \rmBtr {6524}
   & \rmBtr {85.91}
   & \rmBtr {-1.85}
   & \rmUkn {1.1032\,\unit{ }}
   & \rmEql {+2.83}
   & \rmUkn {746\,116}
   & \rmBtr {-1.79}
   & \rmUkn {\formattime{00}{03}{55}}
   & \rmWrs {+6.33}
  \\

     {}
   & {PD}
   & \rmUkn {66500}
   & \rmBtr {-88}
   & \rmBtr {17}
   & \rmUkn {-22\,102}
   & \rmBtr {6742}
   & \rmUkn {86.56}
   & \rmBtr {-1.20}
   & \rmUkn {1.0937\,\unit{ }}
   & \rmEql {+1.95}
   & \rmUkn {733\,950}
   & \rmBtr {-3.40}
   & \rmUkn {\formattime{00}{03}{31}}
   & \rmBtr {-4.52}
  \\

     {}
   & {CD}
   & \rmUkn {66500}
   & \rmUkn {-97}
   & \rmBtr {8}
   & \rmBtr {-22\,022}
   & \rmBtr {6822}
   & \rmUkn {86.07}
   & \rmBtr {-1.69}
   & \rmUkn {1.1030\,\unit{ }}
   & \rmEql {+2.82}
   & \rmBtr {728\,092}
   & \rmBtr {-4.17}
   & \rmUkn {\formattime{00}{03}{50}}
   & \rmWrs {+4.07}
  \\

  \midrule

     {c3}
   & {$L_1$}
   & \rmUkn {286619}
   & \rmUkn {-450}
   & {}
   & \rmUkn {-2\,751\,437}
   & {}
   & \rmUkn {92.33}
   & {}
   & \rmBtr {1.9731\,\unit{ }}
   & {}
   & \rmUkn {2\,131\,502}
   & {}
   & \rmUkn {\formattime{00}{17}{27}}
   & {}
  \\

     {}
   & {SL}
   & \rmUkn {286619}
   & \rmUkn {-315}
   & \rmBtr {135}
   & \rmUkn {-2\,129\,700}
   & \rmBtr {621738}
   & \rmUkn {92.77}
   & \rmEql {0.44}
   & \rmUkn {2.0701\,\unit{ }}
   & \rmEql {+4.92}
   & \rmUkn {2\,107\,934}
   & \rmBtr {-1.11}
   & \rmUkn {\formattime{00}{20}{58}}
   & \rmWrs {+20.15}
  \\

     {}
   & {PD}
   & \rmUkn {286619}
   & \rmUkn {-315}
   & \rmBtr {135}
   & \rmBtr {-1\,831\,312}
   & \rmBtr {920126}
   & \rmUkn {92.81}
   & \rmEql {0.48}
   & \rmUkn {2.0624\,\unit{ }}
   & \rmEql {+4.53}
   & \rmUkn {2\,094\,655}
   & \rmBtr {-1.73}
   & \rmUkn {\formattime{00}{21}{02}}
   & \rmWrs {+20.53}
  \\

     {}
   & {CD}
   & \rmUkn {286619}
   & \rmBtr {-304}
   & \rmBtr {145}
   & \rmUkn {-1\,864\,826}
   & \rmBtr {886611}
   & \rmBtr {91.49}
   & \rmBtr {-0.84}
   & \rmUkn {2.0496\,\unit{ }}
   & \rmEql {+3.88}
   & \rmBtr {2\,085\,304}
   & \rmBtr {-2.17}
   & \rmUkn {\formattime{00}{17}{30}}
   & \rmEql {+0.29}
  \\

  \midrule

     {c4}
   & {$L_1$}
   & \rmUkn {305094}
   & \rmUkn {-84}
   & {}
   & \rmUkn {-147\,690}
   & {}
   & \rmUkn {90.95}
   & {}
   & \rmBtr {6.5049\,\unit{ }}
   & {}
   & \rmUkn {3\,645\,350}
   & {}
   & \rmUkn {\formattime{00}{31}{02}}
   & {}
  \\

     {}
   & {SL}
   & \rmUkn {305094}
   & \rmUkn {-75}
   & \rmBtr {9}
   & \rmUkn {-135\,682}
   & \rmBtr {12008}
   & \rmUkn {90.81}
   & \rmEql {-0.14}
   & \rmUkn {6.5351\,\unit{ }}
   & \rmEql {+0.46}
   & \rmUkn {3\,612\,175}
   & \rmEql {-0.91}
   & \rmUkn {\formattime{00}{44}{12}}
   & \rmWrs {+42.43}
  \\

     {}
   & {PD}
   & \rmUkn {305094}
   & \rmUkn {-65}
   & \rmBtr {19}
   & \rmUkn {-101\,879}
   & \rmBtr {45811}
   & \rmUkn {90.73}
   & \rmEql {-0.22}
   & \rmUkn {6.6586\,\unit{ }}
   & \rmEql {+2.36}
   & \rmUkn {3\,533\,707}
   & \rmBtr {-3.06}
   & \rmUkn {\formattime{00}{34}{26}}
   & \rmWrs {+10.96}
  \\

     {}
   & {CD}
   & \rmUkn {305094}
   & \rmBtr {-64}
   & \rmBtr {20}
   & \rmBtr {-91\,618}
   & \rmBtr {56072}
   & \rmBtr {90.62}
   & \rmEql {-0.33}
   & \rmUkn {6.7194\,\unit{ }}
   & \rmEql {+3.30}
   & \rmBtr {3\,518\,335}
   & \rmBtr {-3.48}
   & \rmUkn {\formattime{00}{32}{09}}
   & \rmWrs {+3.60}
  \\

  \midrule

     {c5}
   & {$L_1$}
   & \rmUkn {420131}
   & \rmUkn {-533}
   & {}
   & \rmUkn {-3\,654\,628}
   & {}
   & \rmUkn {87.49}
   & {}
   & \rmBtr {4.3901\,\unit{ }}
   & {}
   & \rmUkn {4\,300\,895}
   & {}
   & \rmUkn {\formattime{00}{23}{04}}
   & {}
  \\

     {}
   & {SL}
   & \rmUkn {420131}
   & \rmBtr {-404}
   & \rmBtr {129}
   & \rmUkn {-2\,358\,983}
   & \rmBtr {1295645}
   & \rmUkn {87.03}
   & \rmEql {-0.46}
   & \rmUkn {4.5889\,\unit{ }}
   & \rmEql {+4.53}
   & \rmUkn {4\,170\,442}
   & \rmBtr {-3.03}
   & \rmUkn {\formattime{00}{32}{23}}
   & \rmWrs {+40.39}
  \\

     {}
   & {PD}
   & \rmUkn {420131}
   & \rmBtr {-404}
   & \rmBtr {129}
   & \rmUkn {-2\,177\,493}
   & \rmBtr {1477136}
   & \rmUkn {86.75}
   & \rmBtr {-0.74}
   & \rmUkn {4.5325\,\unit{ }}
   & \rmEql {+3.24}
   & \rmUkn {4\,143\,083}
   & \rmBtr {-3.67}
   & \rmUkn {\formattime{00}{25}{34}}
   & \rmWrs {+10.84}
  \\

     {}
   & {CD}
   & \rmUkn {420131}
   & \rmUkn {-415}
   & \rmBtr {118}
   & \rmBtr {-2\,141\,300}
   & \rmBtr {1513328}
   & \rmBtr {86.48}
   & \rmBtr {-1.01}
   & \rmUkn {4.5952\,\unit{ }}
   & \rmEql {+4.67}
   & \rmBtr {4\,093\,042}
   & \rmBtr {-4.83}
   & \rmUkn {\formattime{00}{24}{21}}
   & \rmWrs {+5.56}
  \\

  \midrule

     {c6}
   & {$L_1$}
   & \rmUkn {590060}
   & \rmUkn {-756}
   & {}
   & \rmUkn {-3\,978\,380}
   & {}
   & \rmUkn {88.12}
   & {}
   & \rmBtr {3.6710\,\unit{ }}
   & {}
   & \rmUkn {4\,348\,194}
   & {}
   & \rmUkn {\formattime{00}{35}{49}}
   & {}
  \\

     {}
   & {SL}
   & \rmUkn {590060}
   & \rmUkn {-722}
   & \rmBtr {34}
   & \rmUkn {-3\,160\,751}
   & \rmBtr {817629}
   & \rmUkn {88.94}
   & \rmWrs {0.82}
   & \rmUkn {3.8832\,\unit{ }}
   & \rmWrs {+5.78}
   & \rmUkn {4\,237\,414}
   & \rmBtr {-2.55}
   & \rmUkn {\formattime{00}{46}{15}}
   & \rmWrs {+29.13}
  \\

     {}
   & {PD}
   & \rmUkn {590060}
   & \rmUkn {-698}
   & \rmBtr {58}
   & \rmUkn {-2\,750\,994}
   & \rmBtr {1227386}
   & \rmUkn {88.54}
   & \rmEql {0.42}
   & \rmUkn {3.8960\,\unit{ }}
   & \rmWrs {+6.13}
   & \rmUkn {4\,185\,778}
   & \rmBtr {-3.74}
   & \rmUkn {\formattime{00}{44}{29}}
   & \rmWrs {+24.20}
  \\

     {}
   & {CD}
   & \rmUkn {590060}
   & \rmBtr {-685}
   & \rmBtr {71}
   & \rmBtr {-2\,703\,373}
   & \rmBtr {1275006}
   & \rmBtr {87.17}
   & \rmBtr {-0.95}
   & \rmUkn {3.8637\,\unit{ }}
   & \rmWrs {+5.25}
   & \rmBtr {4\,169\,492}
   & \rmBtr {-4.11}
   & \rmUkn {\formattime{00}{39}{26}}
   & \rmWrs {+10.10}
  \\

  \midrule

     {c7}
   & {$L_1$}
   & \rmUkn {650127}
   & \rmUkn {-144}
   & {}
   & \rmUkn {-339\,184}
   & {}
   & \rmUkn {88.80}
   & {}
   & \rmBtr {13.0465\,\unit{ }}
   & {}
   & \rmUkn {6\,779\,085}
   & {}
   & \rmUkn {\formattime{01}{14}{58}}
   & {}
  \\

     {}
   & {SL}
   & \rmUkn {650127}
   & \rmUkn {-123}
   & \rmBtr {20}
   & \rmUkn {-267\,361}
   & \rmBtr {71823}
   & \rmUkn {88.92}
   & \rmEql {0.12}
   & \rmUkn {13.0961\,\unit{ }}
   & \rmEql {+0.38}
   & \rmUkn {6\,740\,822}
   & \rmEql {-0.56}
   & \rmUkn {\formattime{01}{10}{15}}
   & \rmBtr {-6.29}
  \\

     {}
   & {PD}
   & \rmUkn {650127}
   & \rmBtr {-119}
   & \rmBtr {24}
   & \rmUkn {-217\,997}
   & \rmBtr {121187}
   & \rmBtr {88.72}
   & \rmEql {-0.08}
   & \rmUkn {13.2082\,\unit{ }}
   & \rmEql {+1.24}
   & \rmUkn {6\,653\,426}
   & \rmBtr {-1.85}
   & \rmUkn {\formattime{01}{11}{42}}
   & \rmBtr {-4.36}
  \\

     {}
   & {CD}
   & \rmUkn {650127}
   & \rmUkn {-120}
   & \rmBtr {24}
   & \rmBtr {-194\,336}
   & \rmBtr {144848}
   & \rmUkn {89.12}
   & \rmEql {0.32}
   & \rmUkn {13.2902\,\unit{ }}
   & \rmEql {+1.87}
   & \rmBtr {6\,647\,930}
   & \rmBtr {-1.93}
   & \rmUkn {\formattime{01}{34}{00}}
   & \rmWrs {+25.39}
  \\

  \midrule

     {c8}
   & {$L_1$}
   & \rmUkn {941271}
   & \rmUkn {-197}
   & {}
   & \rmUkn {-1\,093\,449}
   & {}
   & \rmUkn {91.46}
   & {}
   & \rmBtr {13.4100\,\unit{ }}
   & {}
   & \rmUkn {11\,450\,165}
   & {}
   & \rmUkn {\formattime{02}{11}{43}}
   & {}
  \\

     {}
   & {SL}
   & \rmUkn {941271}
   & \rmUkn {-154}
   & \rmBtr {43}
   & \rmUkn {-985\,353}
   & \rmBtr {108096}
   & \rmUkn {91.59}
   & \rmEql {0.13}
   & \rmUkn {13.5111\,\unit{ }}
   & \rmEql {+0.75}
   & \rmUkn {11\,405\,990}
   & \rmEql {-0.39}
   & \rmUkn {\formattime{02}{54}{40}}
   & \rmWrs {+32.61}
  \\

     {}
   & {PD}
   & \rmUkn {941271}
   & \rmBtr {-114}
   & \rmBtr {83}
   & \rmUkn {-804\,198}
   & \rmBtr {289251}
   & \rmUkn {91.15}
   & \rmEql {-0.31}
   & \rmUkn {13.8756\,\unit{ }}
   & \rmEql {+3.47}
   & \rmUkn {11\,196\,187}
   & \rmBtr {-2.22}
   & \rmUkn {\formattime{02}{44}{06}}
   & \rmWrs {+24.59}
  \\

     {}
   & {CD}
   & \rmUkn {941271}
   & \rmUkn {-115}
   & \rmBtr {82}
   & \rmBtr {-804\,080}
   & \rmBtr {289369}
   & \rmBtr {91.14}
   & \rmEql {-0.32}
   & \rmUkn {14.1136\,\unit{ }}
   & \rmWrs {+5.25}
   & \rmBtr {11\,138\,570}
   & \rmBtr {-2.72}
   & \rmUkn {\formattime{02}{31}{19}}
   & \rmWrs {+14.88}
  \\



  \midrule
     {all}
   & {$L_1$}
   & {}
   & \rmUkn {-2\,353}
   & {}
   & \rmUkn {-12\,003\,449}
   & {}
   & \rmUkn {89.88}
   & {}
   & \rmBtr {44.6564\,\unit{ }}
   & {}
   & \rmUkn {33\,951\,430}
   & {}
   & \rmUkn {\formattime{05}{20}{58}}
   & {}
  \\

     {}
   & {SL}
   & {}
   & \rmUkn {-1\,949}
   & {}
   & \rmUkn {-9\,066\,474}
   & {}
   & \rmUkn {89.79}
   & \rmEql {-0.10}
   & \rmUkn {45.3916\,\unit{ }}
   & \rmEql {+1.65}
   & \rmUkn {33\,557\,291}
   & \rmBtr {-1.16}
   & \rmUkn {\formattime{06}{35}{49}}
   & \rmWrs {+23.32}
  \\

     {}
   & {PD}
   & {}
   & \rmUkn {-1\,867}
   & {}
   & \rmUkn {-7\,911\,912}
   & {}
   & \rmUkn {89.68}
   & \rmEql {-0.20}
   & \rmUkn {45.9321\,\unit{ }}
   & \rmEql {+2.86}
   & \rmUkn {33\,071\,485}
   & \rmBtr {-2.59}
   & \rmUkn {\formattime{06}{08}{24}}
   & \rmWrs {+14.78}
  \\

     {}
   & {CD}
   & {}
   & \rmBtr {-1\,861}
   & {}
   & \rmBtr {-7\,827\,888}
   & {}
   & \rmBtr {89.25}
   & \rmBtr {-0.63}
   & \rmUkn {46.3456\,\unit{ }}
   & \rmEql {+3.78}
   & \rmBtr {32\,907\,192}
   & \rmBtr {-3.08}
   & \rmUkn {\formattime{06}{05}{44}}
   & \rmWrs {+13.95}
  \\


  \bottomrule
\end{tabular}

\end{table}

\section{Conclusion}
We presented a new and faster approximation algorithm for the cost-distance Steiner tree problem, achieving the best known
theoretical approximation factor of $\mathcal{O}(\log t)$.
It arises as a sub-problem in timing-constrained global routing with a linear delay model (before buffering).
We adopted it for better practical performance.
Our experimental results show improvements over state-of-the-art algorithms on large instances and in presence of bifurcation penalties.
It leads to the overall best timing and congestion results in timing-constrained global routing.

\section*{Acknowledgements}
We thank Paula Heinz who implemented a preliminary version based on \cite{meyerson-etal:08}.
We reused parts her code.

\end{document}